\documentclass[journal]{IEEEtran}
\usepackage{amssymb}
\usepackage{amsfonts}
\usepackage{bbm}
\usepackage{amsfonts}
\usepackage[dvips]{graphicx}
\usepackage[dvips]{color}
\usepackage{graphicx}
\usepackage{amsmath}
\usepackage{fancyhdr}
\usepackage{stfloats}
\usepackage{multirow}
\usepackage{algorithm}
\usepackage{algorithmic}
\usepackage{cite}
\usepackage[bookmarks=false]{}
\usepackage{amsthm}
\usepackage{algorithm}
\usepackage{algorithmic}
\usepackage{mathbbol}
\usepackage{amsfonts}
\usepackage{graphicx}
\usepackage{amsmath}
\usepackage{amssymb}
\usepackage{latexsym}
\usepackage{graphicx}
\usepackage{cite}
\usepackage{url}
\usepackage{stfloats}
\usepackage{mathrsfs}
\usepackage{setspace}
\usepackage{booktabs}
\usepackage{latexsym}
\usepackage{url}
\usepackage{stfloats}
\usepackage{mathrsfs}
\theoremstyle{plain}

\newtheorem{lemma}{Lemma}

\usepackage{cite}
\usepackage{setspace}
\usepackage{textcomp}
\usepackage{xcolor}
\usepackage{makecell}
\usepackage{etoolbox}
\usepackage{diagbox}
\usepackage{caption}
\usepackage{amsthm}
\newtheorem*{remark}{Remark}

\usepackage[tight,footnotesize]{subfigure}
\begin{document}
\clearpage
\title{\huge Bandwidth and {P}ower {A}llocation for {T}ask-{O}riented {S}emantic {C}ommunication}
%
\author{\IEEEauthorblockN{Chuanhong Liu\IEEEauthorrefmark{1}, Caili Guo\IEEEauthorrefmark{1}, Yang Yang\IEEEauthorrefmark{2}, and Jiujiu Chen\IEEEauthorrefmark{1}
\\\IEEEauthorrefmark{1}Beijing Key Laboratory of Network System Architecture and Convergence, School of
Information and Communication Engineering, Beijing University of Posts and Telecommunications, Beijing 100876, China
\\\IEEEauthorrefmark{2}Beijing Laboratory of Advanced Information Networks, School of
Information and Communication Engineering, Beijing University of Posts and Telecommunications, Beijing 100876, China
\\Email:\{2016\_liuchuanhong,guocaili,yangyang01,chenjiujiu\}@bupt.edu.cn}
\thanks{ This work was supported by National Natural Science Foundation of China (61871047), National Natural Science Foundation of China (61901047), Beijing Natural Science Foundation (4204106) and China Postdoctoral Science Foundation (2018M641278).
}
}
%
%
%
\maketitle
\pagestyle{empty}
\vspace{0cm}
\begin{abstract}
Deep learning enabled semantic communication has been studied to improve  communication efficiency while guaranteeing intelligent task performance. 
Different from conventional communications systems, the resource allocation in semantic communications no longer just pursues the bit transmission rate, but focuses on how to better compress and transmit semantic to complete subsequent intelligent tasks. This paper aims to appropriately allocate the bandwidth and power for artificial intelligence (AI) task-oriented semantic communication and proposes a joint compressiom ratio and resource allocation (CRRA) algorithm. We first analyze the relationship between the AI task's performance and the semantic information. Then, to optimize the AI task's perfomance under resource constraints, a bandwidth and power allocation problem is formulated. The problem is first separated into two subproblems due to the non-convexity. The first subproblem is a compression ratio optimization problem with a given resource allocation scheme, which is solved by a enumeration algorithm. The second subproblem is to find the optimal resource allocation scheme, which is transformed into a convex problem by successive convex approximation method, and solved by a convex optimization method. The optimal semantic compression ratio and resource allocation scheme are obtained by iteratively solving these two subproblems. Simulation results show that the proposed algorithm can efficiently improve the AI task's performance by up to 30\% comprared with baselines.
\end{abstract}

\vspace{-0.3cm}
\section{Introduction}
\label{sec:intro}

\IEEEPARstart{R}ecently, semantic communication has attracted extensive attention from industrial and academia\cite{Nine}, which have been identified as a core challenge for the sixth generation (6G) of wireless networks. Semantic communications only transmit necessary information relevant to the specific task at the receiver\cite{info}, which leads to a truly intelligent system with significant reduction in data traffic\cite{Qin_survey}.

There are several priori studies on semantic communications based on deep learning. For text data, the authors in \cite{Xie_Deep} proposed a semantic communication system based on Transformer, in which the concept of semantic information was clarified at the sentence level. Based on \cite{Xie_Deep}, the authors in \cite{Xie_lite} further proposed a lite disributed semantic communication system, making the model easier to deploy at the Internet of things (IoT) devices. For image data, the authors in \cite{Gunduz_JSCC} propoesd a joint source-channel coding scheme based on convolutional neural network (CNN) to transmit image data in wireless channel, which can jointly optimize various modules of communication system. 
In our previous work \cite{Liu_AIoT}, we have proposed an intelligent task-oriented semantic communication method in Artificial Intelligence \& Internet of Things (AIoT), in which semantic compression is proposed based on semantic relationship between concepts and feature maps. The expriment results show that the redundancy can be further removed and the task perfomance can be improved, especially in a resource-constrained environment. Above all, the existing works on semantic communication are focused on the implementation of semantic compression and semantic communication systems. However, appropriate resource allocation is also significant to semantic communications. In particular, semantic encoding is essentially a compression task, which is closely related to the allocated bandwidth resources. In addition, power is a key resource in AIoT scenarios, which is closely related to the life of the IoT devices.
To appropriately optimize bandwidth, power allocation and semantic compression ratio in semantic communications, following two key issues remain to be solved:

\emph{Question 1: how to determine the compression ratio to obtain the best trade-off between task performance and data transmission?}

\emph{Question 2: how to appropriately allocate communication resources (including bandwidth and transmit power) for multiple devices to obtain global optimal performance?}

In this paper, we investigate the optimal semantic compression ratio and the optimal resource allocation method in the task-oriented semantic communication system. \textit{To our best knowledge, this is the first work that study resource allocation in semantic communications}. The main contributions of this paper are summarized as follows:

\begin{itemize}
	\item[$\bullet$] We first study the relationship between AI task performance and semantic compression ratio. For image classification tasks, we use the curve fitting method to obtain the specific relationship between the performance and the semantic compression ratio. Based on the relationship, a joint semantic compression ratio and resource allocation problem is formulated, whose goal is to maximize the task's performance.
	\item[$\bullet$] To solve the problem, we first separate it into two sub-problems. The first subproblem is a compression ratio optimization problem with a given resource allocation scheme, which is solved by a enumeration algorithm. This addresses the aforementioned \emph{Question 1}.
	\item[$\bullet$] Based on the obtained compression ratio, the second subproblem is to find the optimal resource allocation, which is turned into a convex problem and solved by a convex optimization method. The two subproblems are then updated iteratively until a convergent solution is obtained. This addresses the aforementioned \emph{Question 2}.
\end{itemize}
	
%

Simulation results show that the proposed algorithm can obtain 30\% gains in terms of the AI task's performance, especially in low resources regime, compared to the algorithm without considering dynamic optimization of compression ratio and resource allocation.

The remainder of this paper is organized as follows. System model and problem formulation are described in Section II. Section III presents the proposed joint semantic compression ratio optimization and resource allocation algorithm. Simulation and numerical results are presented and discussed in Section IV. Finally, Section V draws some important conclusions.


\vspace{-0.cm}
\section{System Model and Problem Formulation}
\label{sec:system}
In this section, we first introduce the task-oriented semantic communication model. Then, the channel model is given. Next, the intelligent task performance model based on semantic compression is derived. Finally, based on the established models, we formulate a joint semantic compression ratio and resource allocation optimization problem to maximize the task performance.

\vspace{-0.3cm}
\subsection{Task-Oriented Semantic Communication Model}
Fig. 1 illustrates the semantic communication model, which consists of a sender, a physical channel, and a receiver. The sender mainly executes semantic coding, channel coding and modulation. The semantic coding consists of the semantic extraction, the semantic relation extraction and the semantic compression. The receiver mainly performs demodulation, channel decoding and intelligent task calculation (e.g., a classifier in image classification task, a detector in target detection task, etc.). 

\begin{figure}[t]
	\begin{center}
		\includegraphics[width=1\linewidth]{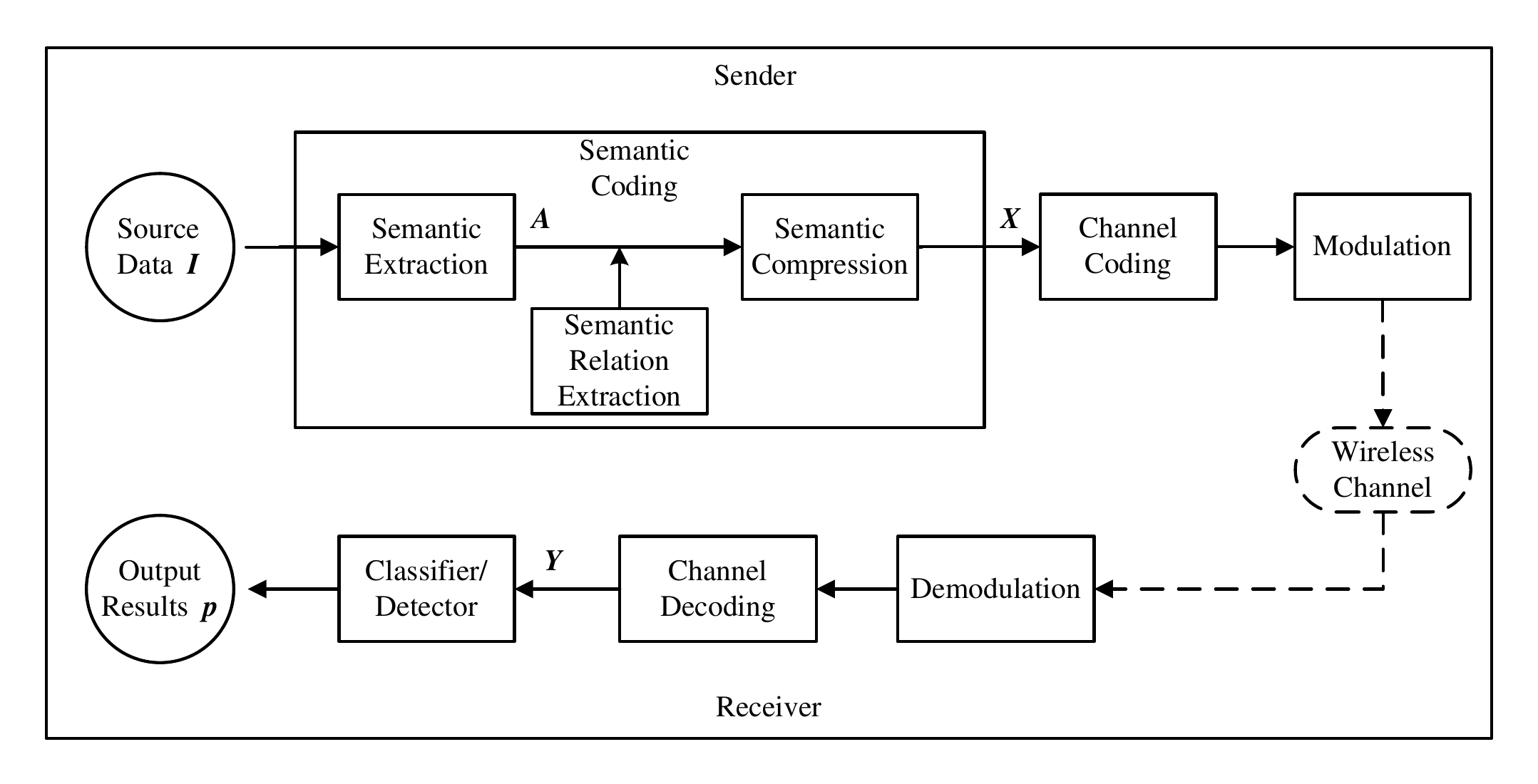}
	\end{center}
	\caption{The system model of semantic communication.}
	\label{fig:system}
\end{figure}

For source data (e.g., images) ${\boldsymbol I}$, the sender first uses neural networks (e.g., CNN) to extract the semantic information, which can be denoted by:
\begin{eqnarray}\label{signalform}
	{\boldsymbol {A}} = {S_{\boldsymbol {\alpha }}}({\boldsymbol {I}}),
\end{eqnarray}
where ${S_{\boldsymbol{\alpha }}}()$ denotes the semantic extraction network, ${\boldsymbol{\alpha }}$ is the parameters of the neural network.

Then the extracted data ${\boldsymbol{A}}$ is compressed based on semantic relation, named sematic compression, which can be expressed as follows:
\begin{eqnarray}\label{}
	{\boldsymbol{X}} = {C_o }({\boldsymbol{A}}),
\end{eqnarray}
where ${C_o }()$ denotes the semantic compression function, $o$ is the compression ratio. Semantic compression can further remove data redundancy from the semantic level. The extracted semantic information can be a series of features. Let the $k$-th feature be ${{\boldsymbol{A}}^{_k}}$. The process of semantic compression can be expressed as
\begin{eqnarray}\label{compression}
	{{\boldsymbol{A}}^{_k}} = \left\{ {\begin{array}{*{20}{c}}
		{{{\boldsymbol{A}}^{_k}},{\omega _k} \ge {\omega _0}}\\
		{0,{\omega _k} < {\omega _0}}
\end{array}} \right.
\end{eqnarray}
where ${\omega _k}$ is the importance weight of the $k$-th feature and ${\omega _0}$ is the compression threshold. Formula (\ref{compression}) indicates that when the importance weight of the feature is greater than the threshold, it will be transmitted, otherwise it will not be transmitted. Thus the compression ratio $o$ can be computed by $o = \frac{{{N_T}}}{N}$, while $N_T$ is the number of transmitted feature, and $N$ is the number of total features.
Semantic compression has two major benefits: first, it reduces the subsequent computing resource requirements; second, it greatly reduces the amount of data transmitted, and reduces the demand for communication resources and transmission delay. It is worth mentioning that our resource allocation algorithm is suitable for different semantic communication models and semantic compression methods. Here we only take the previously proposed task-oriented semantic communication model as an example to illustrate\cite{Liu_AIoT}, in which the compression unit is feature map.

Then the compressed semantic information is transmitted through the physical channel, and the received semantic information is:
\begin{eqnarray}\label{}
{\boldsymbol{Y}} = h{\boldsymbol{X}} + n,
\end{eqnarray}
where $h$ denotes the channel gain, and $n$ is the power of additive white Gaussian noise (AWGN).

Finally, the receiver directly inputs received semantic information $\boldsymbol{Y}$ to neural network to complete the intelligent tasks. Taking the image classification task as an example, $\boldsymbol{Y}$ is input into the classifier and the output is the probability corresponding to each class, which can be denoted by:
\begin{eqnarray}\label{}
	{\boldsymbol{p}} = {Q_{\boldsymbol{\mu }}}({\bf{Y}}),
\end{eqnarray}
where ${\boldsymbol{p}} = [{p_1},p{}_2,...,{p_M}]$, ${p_k}$ is the probability that the image is classified into the $k$-th class, $M$ is the total number of classes, ${Q_{\boldsymbol{\mu }}}()$ denotes the classifier network, and ${\boldsymbol{\mu }}$ is the parameter.

\vspace{-0.3cm}
\subsection{Transmission Model} 
Consider a cellular network consisting of a set ${{\cal U}}$ of $U$ users and one edge sever. Users extract semantic information locally and then transmit them to edge sever to complete AI tasks.
Transmission in pyhsical layer still follows Shanno's classic information theory, and the transmission rate of user $n$ is
\begin{eqnarray}\label{R}
	{R_i} = {B_i}{\rm{lo}}{{\rm{g}}_{\rm{2}}}{\rm{(1 + }}\frac{{{h_i}{P_i}}}{{{N_0}{B_i}}}{\rm{)}},
\end{eqnarray}
where $B_i$ is bandwidth of user $i$, $P_i$ is the transmission power of user $i$, $h_i$ is the channel gain between user $i$ and edge sever, and $N_0$ is the noise power spectral density.

Denoting the initial data size of semantic information that users need to upload is $d_0$, and the semantic compression ratio of user $i$ is $o_i$, then the data size of the user $i$ uploading to edge sever is ${d_i} = {d_0} \times {o_i}$. Therefore, the transmission delay of user $i$ is 
\begin{eqnarray}\label{t}
{t_i} = \frac{{{d_i}}}{{{R_i}}}.
\end{eqnarray}

In actual scenarios, there is always a strict delay threshold for completing the AI task, which can be denoted by $t_0$. Let the success transmission probability of user $i$ be ${\rm{P}}({t_i} \le {t_0})$. To calculate ${\rm{P}}({t_i} \le {t_0})$, we have the following lemma.

\begin{lemma}
The success transmission probability of user $i$ is
\begin{align}\label{transmission}
	P({t_i} \le {t_0}) = 2Q\left( {\frac{{{2^{{a_i}(1 - {o_i})}} - 1}}{{{b_i}\delta }}} \right),
\end{align}
where ${a_i} = \frac{{{d_0}}}{{{B_i}{t_0}}}$, ${b_i} = \frac{{{P_i}}}{{{N_0}{B_i}}}$ and ${\delta ^2}$ is variance of the channel state information. The $Q$-function is the tail distribution function of the standard normal distribution.
\end{lemma}
\begin{proof}
	Based on (\ref{R}) and (\ref{t}), we have
\begin{align}\label{proof}
	P({t_i} \le {t_0}) = &{\rm{P}}\left( {\frac{{\left( {1 - {o_i}} \right){d_0}}}{{{B_i}{{\log }_2}\left( {1 + \frac{{{h_i}{P_i}}}{{{N_0}{B_i}}}} \right)}} \le {t_0}} \right)\\
	= &{\rm{P}}\left( {\frac{{{2^{{a_i}(1 - {o_i})}} - 1}}{{{b_i}}} \le {h_i}} \right)\tag{\theequation a}\\
	= &2Q\left( {\frac{{{2^{{a_i}(1 - {o_i})}} - 1}}{{{b_i}\delta }}} \right)\tag{\theequation b},
\end{align}
where (\ref{proof}b) follows from ${h_i} \sim N(0,{\delta ^2})$.

This ends the proof.
\end{proof}

\begin{remark}
	As we have seen above, the success transmission probability is mainly affected by power, bandwidth and semantic compression ratio.
\end{remark}

\vspace{-0.3cm}
\subsection{Intelligent Task Performance Model}
There is an implicit relationship between the semantic compression ratio and the performance of AI task ($i.e.$ accuracy for image classification task).  Specifically, the higher the semantic compression ratio, the worse the performance of the final task, and vice versa. In this subsection, we aim to explore the mathematical relationship between AI task performance and semantic compression ratio. 


\begin{figure*}[htbp]
	\begin{center}
		\includegraphics[height=0.35\textheight,width=1\linewidth]{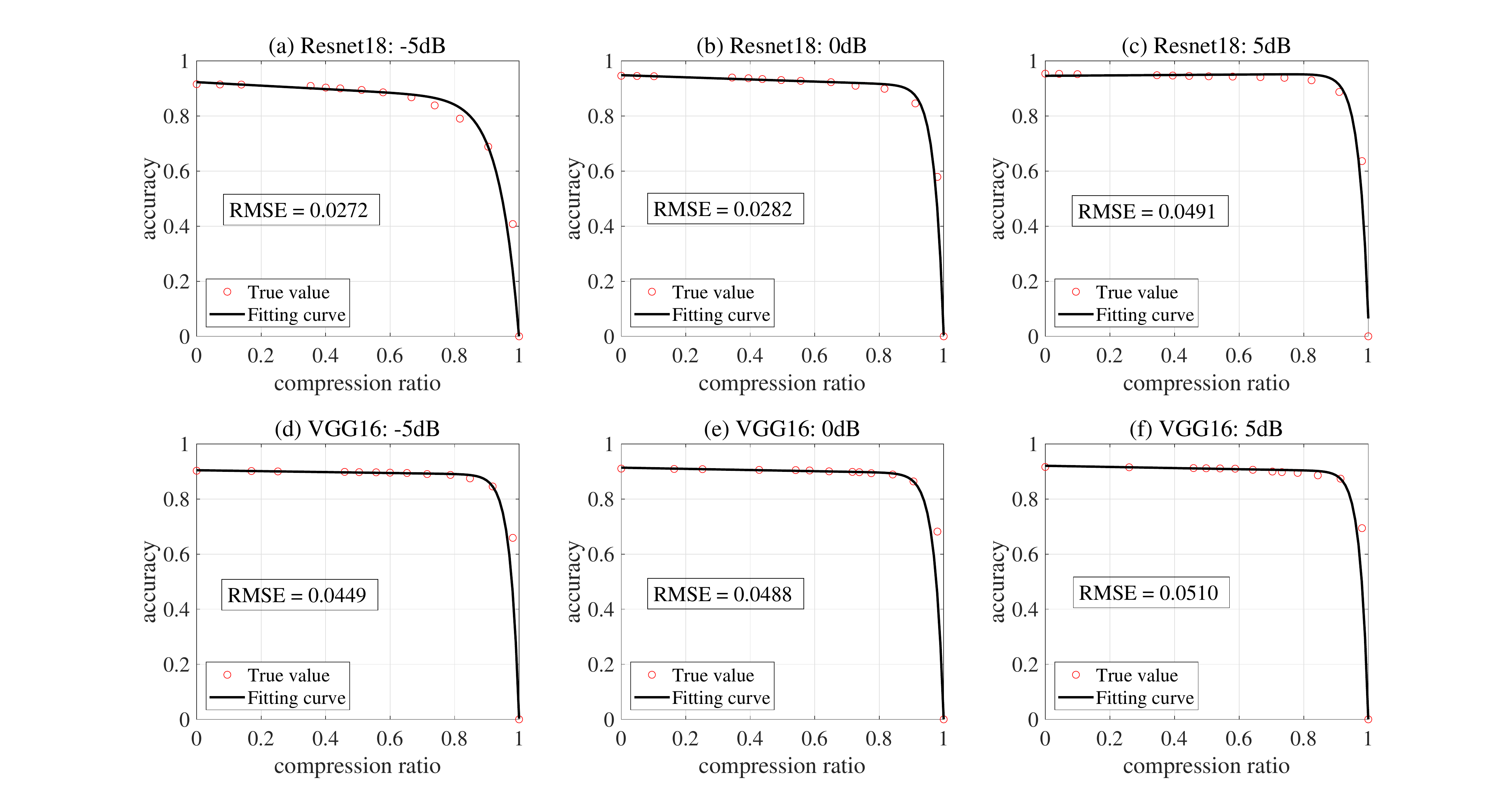}
	\end{center}
	\caption{Relationship between accuracy and compression ratio of different convolutional neural networks. For each subfigure, the x-axis represents the compression ratio and the y-axis is the accuracy.}
	\label{fig:fitting}
\end{figure*}

To obtain the accuracy under different compression ratios ($i.e.$, a point set ${{\cal D}}$ of $D$ points), we first calculate the importance ranking of the feature maps \cite{Liu_AIoT}, then remove the unimportant feature maps in turn and calculate the corresponding classification accuracy and compression ratio. As shown in Fig. \ref{fig:fitting}, the accuracy under different compression ratios of different neural networks are possibly variable. The models in Fig. \ref{fig:fitting}(a) - (c) are based on Resnet18, and retrained in different channel SNR ($i.e.$, -5 dB, 0 dB and 5 dB) while the models in Fig. \ref{fig:fitting}(d) - (f) are based on VGG16. Inspired by the ideas in \cite{fitting}, we empirically find that the curve of ${{\cal D}}$ can be estimated by an exponential function: ${\eta}(o) = {\beta _1}{e^{{\beta _2}o}} + {\beta _3}{e^{{\beta _4}o}}$. Then, we fit the function to obtain the parameters ${\boldsymbol{\beta }} = [{\beta _1},{\beta _2},{\beta _3},{\beta _4}]$, which are different for different neural network, based on gradient descent method. Also in Fig. \ref{fig:fitting}, we find that the curve of ${{\cal D}}$ is well approximated by our exponential function with extremely small reconstruction error, which is quantified by root-mean-square error (RMSE). 

\vspace{-0.5cm}
\subsection{Performance Metric}
In task-oriented semantic communication, considering both transmission and intelligence task performance, we initialize a new metric, named effective accuracy.
We denote the classification at receiver as event A, and the transmission as event B. ${\rm{P}}(A=1)$ means the classification at receiver is correct and ${\rm{P}}(B=1)$ means semantic information are transmitted succcessfully, and otherwise, we have ${\rm{P}}(B=0)$. Thus, our goal is equivalent to maximize the value of ${\rm{P}}(A=1)$. According to the law of total probability, we have
\begin{eqnarray}\label{lawtotal}
\begin{aligned}
{\rm{P}}\left( {A = 1} \right) = {\rm{P}}\left( {A = 1\left| {B = 1} \right.} \right) \times {\rm{P}}\left( {B = 1} \right) \\+ {\rm{P}}\left( {A = 1\left| {B = 0} \right.} \right) \times {\rm{P}}\left( {B = 0} \right).
\end{aligned}
\end{eqnarray}
Since the second term on the right side of the equation means the probability of a successful classification under failure transmission which equals to 0, (\ref{lawtotal}) can be simplified to ${\rm{P}}\left( {A = 1} \right) = {\rm{P}}\left( {A = 1\left| {B = 1} \right.} \right) \times {\rm{P}}\left( {B = 1} \right)$. In the semantic communication system, ${\rm{P}}\left( {B = 1} \right)$ is equivalent to the probability that the transmission delay is less than the delay constraint while ${\rm{P}}\left( {A = 1\left| {B = 1} \right.} \right)$ is equivalent to the classification accuracy after semantic compression. 

Above all, in order to simultaneously evaluate the impact of communication transmission and semantic compression on the performance of AI tasks, the effective classification accuracy of the classification task of user $i$ can be expressed as
\begin{eqnarray}\label{}
{\gamma _i} = {\rm{P}}({t_i} \le {t_0}) \times {\eta }({o_i}).
\end{eqnarray}

\vspace{-0.3cm}
\subsection{Probelm Formulation}
We aim to jointly optimize the semantic compression ratio and resource allocation for all users to maximize the total effective accuracy, which is given as
\begin{align}\label{Q1}
	&\mathop {\max }\limits_{{\boldsymbol{B}},{\boldsymbol{P}},{\boldsymbol{o}}} \sum\limits_{i = 1}^U {{\gamma _i}} \\
	\rm{s.t.}\;\;\;&{B_i} \ge {B_{\min }},\tag{\theequation a}\\
	&\sum\limits_{i = 1}^U {{B_i}}  \le {B_{\max }},\tag{\theequation b}\\
	&{P_i} \ge {P_{\min }},\tag{\theequation c}\\
	&\sum\limits_{i = 1}^U {{P_i}}  \le {P_{\max }},\tag{\theequation d}\\
	&0 < {o_i} < 1,\tag{\theequation f}
\end{align}
where $B_{\min }$ is the minimum bandwidth allocated to users, $B_{\max }$ is the maximum total bandwidth of all users, $P_{\min }$ is the minimum transmit power allocated to users, and $P_{\max }$ is the maximum total transmit power of all users. Constraint (\ref{Q1}a) denotes the constraint of minimum bandwidth. Constraint (\ref{Q1}b) shows the overall bandwidth constraint. Constraint (\ref{Q1}c) denotes the constraint of minimum power. Constraint (\ref{Q1}d) shows the overall power constraint. Constraint (\ref{Q1}f) means that the semantic compression ratio is between 0-1.



\vspace{-0.3cm}
\section{Proposed Algorithm}
\label{sec:algorithm}

This section proposes a joint compressiom ratio and resource allocation (CRRA) algorithm to solve (\ref{Q1}).

Utilizing the approximation form of $Q$-function $Q(x) \le \frac{1}{2}{e^{ - \frac{{{x^2}}}{2}}}$ \cite{Q_bound}. Problem (\ref{Q1}) can be simplified as 

\begin{align}\label{Q3}
	&\mathop {\max }\limits_{{\boldsymbol{B}},{\boldsymbol{P}},{\boldsymbol{o}}} \sum\limits_{i = 1}^U {{e^{ - \frac{1}{2}{{\left[ {\frac{{{N_0}{B_i}\left[ {{2^{\left[ {\frac{{{d_0}\left( {1 - {o_i}} \right)}}{{{B_i}{t_0}}}} \right]}} - 1} \right]}}{{\delta {P_i}}}} \right]}^2}}} \times } {\eta }({o_i})\\
	\rm{s.t.}\;\;\;&(\ref{Q1}a)-(\ref{Q1}f)\nonumber.
\end{align}

Since problem (\ref{Q3}) is still non-convex, we first divide (\ref{Q3}) into two subproblems, and then solve these two subproblems iteratively. In particular, we first fix the resource allocation and calculate the optimal compression ratio for each user. Then, the problem of resource allocation is formulated and solved with the obtained compression ratio. The two subproblems are iteratively solved until a convergent solution is obtained. 

\vspace{-0.3cm}
\subsection{Optimal Compression Ratio}
Given the bandwidth and power allocation, let 
\begin{align}
	h(o_i) = {e^{ - \frac{1}{2}{{\left[ {\frac{{{N_0}{B_i}\left[ {{2^{\left[ {\frac{{{d_0}\left( {1 - {o_i}} \right)}}{{{B_i}{t_0}}}} \right]}} - 1} \right]}}{{\delta {P_i}}}} \right]}^2}}} \times {\eta }({o_i})
\end{align}
then (\ref{Q3}) can be simplified as 
\begin{align}\label{Q4}
	&\mathop {\max }\limits_{{\boldsymbol{o}}} \sum\limits_{i = 1}^U h(o_i)\\
	\rm{s.t.}\;\;\;&0 < {o_i} < 1.\tag{\theequation a}
\end{align}

We can observe from (\ref{Q4}) that once the bandwidth and power allocation are fixed, the optimal semantic compression ratio of each user is independent. Thus, our goal transforms into maximizing each user's effective accuracy by optimizing the semantic compression, which means there is no need to consider cumulative sums. For user $i$, the problem is 
\begin{align}\label{Q5}
&\mathop {\max }\limits_{{o_i}} h(o_i)\\
\rm{s.t.}\;\;\;&0 < {o_i} < 1.\tag{\theequation a}
\end{align}

Considering the range of $o_i$ is between 0 and 1, and the compression unit is feature map, we here employ the one-dimension enumeration method to obtain the optimal semantic compression ratio.



\vspace{-0.3cm}
\subsection{Optimal Resource Allocation}
With the obtained semantic compresseion ratios, we then optimize the bandwidth and power of the considered semantic communication systems. Note that given $o_i$, ${{\eta}\left( {{o_i}} \right)}$ can be seen as a constant, which is denoted as ${\alpha _i}$. Thus, the resource allocation probelm can be reformulated as
\begin{align}\label{Q6}
	&\mathop {\min }\limits_{{\boldsymbol{B}},{\boldsymbol{P}}} \sum\limits_{i = 1}^U { - {\alpha _i} \times {e^{ - \frac{1}{2}{{\left\{ {\frac{{{N_0}{B_i}\left[ {{2^{\frac{{{d_0}\left( {1 - \sigma } \right)}}{{{B_i}{t_i}}}}} - 1} \right]}}{{\delta {P_i}}}} \right\}}^2}}}} \\
	\rm{s.t.}\;\;\;&(\ref{Q1}a)-(\ref{Q1}d)\nonumber.
\end{align}

To solve problem (\ref{Q6}), we first convert the non-convex problem into a convex optimization problem. In particular, by introducing slack variables ${\boldsymbol{f}} = [{f_1},{f_2},...,{f_U}]$, ${\boldsymbol{y}} = [{y_1},{y_2},...,{y_U}]$, ${\boldsymbol{x}} = [{x_1},{x_2},...,{x_U}]$, ${\boldsymbol{m}} = [{m_1},{m_2},...,{m_U}]$ and ${\boldsymbol{q}} = [{q_1},{q_2},...,{q_U}]$, problem (\ref{Q6}) can be transformed into
\begin{align}\label{Q7}
	&\mathop {\min }\limits_{{\boldsymbol{B}},{\boldsymbol{P}},{\boldsymbol{f}},{\boldsymbol{y}},{\boldsymbol{x}},{\boldsymbol{m}},{\boldsymbol{q}}} \sum\limits_{i = 1}^U { - {\alpha _i} \times {f_i}} \\
	\rm{s.t.}{\rm{    }}&{f_i} \le {e^{{y_i}}},i = 1,2,...,U, \tag{\theequation a}\\
	{\rm{        }}&{y_i} \le  - \frac{1}{2}x_i^2,i = 1,2,...,U, \tag{\theequation b}\\
	{\rm{        }}&{x_i} \ge \frac{{{N_0}{B_i}{m_i}}}{{\delta {P_i}}}, \tag{\theequation c}\\
	{\rm{       }}&{m_i} \ge {2^{{q_i}}} - 1, \tag{\theequation d}\\
	{\rm{       }}&{q_i} \ge \frac{{{d_0}\left( {1 - \sigma } \right)}}{{{B_i}{t_i}}}, \tag{\theequation e}\\
	{\rm{       }}&(\ref{Q1}a)-(\ref{Q1}d).\nonumber
\end{align}
However, constraints (\ref{Q7}a) and (\ref{Q7}c) are still non-convex. 

For constraint (\ref{Q7}a), we use the successive convex approximation (SCA) method to turn it into a convex constraint. Performing a first-order Taylor expansion of ${e^{{y_i}}}$ at ${e^{{y_i}^j}}$, then we have
\begin{eqnarray}\label{}
{f_i} \le {e^{y_i^j}} + \left( {{y_i} - y_i^j} \right){e^{y_i^j}}, 
\end{eqnarray}
where the superscript $j$ represents the value obtained after $j$-th iteration of the variable. 

For constraint (\ref{Q7}c), slack variable ${\boldsymbol{z}} = [{z_1},{z_2},...,{z_U}]$ is introduced, and have 
\begin{eqnarray}\label{z}
{z_i} \ge {B_i}{m_i}.
\end{eqnarray} 
Thus, constraint (\ref{Q7}c) can be transformed into 
\begin{eqnarray}\label{xp}
{x_i}{P_i} \ge \frac{{{N_0}{z_i}}}{{{\delta _i}}}.
\end{eqnarray}
(\ref{z}) can be rewritten as
\begin{eqnarray}\label{}
{z_i} \ge {B_i}{m_i} = \frac{1}{4}\left( {{{\left( {{B_i} + {m_i}} \right)}^2} - {{\left( {{B_i} - {m_i}} \right)}^2}} \right).\end{eqnarray}
By performing a first-order Taylor expansion of ${\left( {{B_i} - {m_i}} \right)^2}$ at point $\left( {B_i^j,m_i^j} \right)$
and using SCA, we have
\begin{eqnarray}\label{z_i1}
\begin{aligned}
{z_i} \ge \frac{1}{4}( {{\left( {{B_i} + {m_i}} \right)}^2} - 2\left( {{B_i} - {m_i}} \right)\left( {{B_i}^j - {m_i}^j} \right) \\ + {{\left( {{B_i}^j - {m_i}^j} \right)}^{\rm{2}}} ). 
\end{aligned}
\end{eqnarray}
Similarly, (\ref{xp}) is equivalent to
\begin{eqnarray}\label{}
 {x_i}{P_i}{\rm{ = }}\frac{{\rm{1}}}{{\rm{4}}}\left( {{{\left( {{x_i} + {P_i}} \right)}^2} - {{\left( {{x_i} - {P_i}} \right)}^2}} \right) \ge \frac{{{N_0}{z_i}}}{{{\delta _i}}}. 
\end{eqnarray}
By performing a first-order Taylor expansion of ${\left( {{x_i} + {P_i}} \right)^2}$ and ${\left( {{x_i} - {P_i}} \right)^2}$ at point $\left( {{x_i}^j,{P_i}^j} \right)$ and using SCA, we can obtain 
\begin{eqnarray}\label{z_i2}
\begin{aligned}
\frac{{4{N_0}{z_i}}}{{{\delta _i}}} \le 2\left( {{x_i} + {P_i}} \right)*\left( {{x_i}^j + {P_i}^j} \right) - {\left( {{x_i}^j + {P_i}^j} \right)^2} \\- 2\left( {{x_i} - {P_i}} \right)*\left( {{x_i}^j - {P_i}^j} \right) + {\left( {{x_i}^j + {P_i}^j} \right)^2}.
\end{aligned}
\end{eqnarray}

So far, all constraints are transformed into convex, and the optimization problem can be reformulated as
\begin{align}\label{Q8}
	&\mathop {\min }\limits_{{\boldsymbol{B}},{\boldsymbol{P}},{\boldsymbol{f}},{\boldsymbol{y}},{\boldsymbol{x}},{\boldsymbol{m}},{\boldsymbol{q}},{\boldsymbol{z}}} \sum\limits_{i = 1}^U { - {\alpha _i} \times {f_i}} \\
	\rm{s.t.}{\rm{    }}&{f_i} \le {e^{y_i^j}} + \left( {{y_i} - y_i^j} \right){e^{y_i^j}},i = 1,2,...,U, \tag{\theequation a}\\
	{\rm{        }}&{y_i} \le  - \frac{1}{2}x_i^2,i = 1,2,...,U, \tag{\theequation b}\\
	{\rm{       }}&{m_i} \ge {2^{{q_i}}} - 1, \tag{\theequation c}\\
	{\rm{       }}&{q_i} \ge \frac{{{d_0}\left( {1 - \sigma } \right)}}{{{B_i}{t_i}}}, \tag{\theequation d}\\
	{\rm{       }}&(\ref{Q1}a)-(\ref{Q1}d), (\ref{z_i1}), (\ref{z_i2}).\nonumber
\end{align}

Problem (\ref{Q8}) is a convex optimization problem, and can be effectively solved via the dual method\cite{convex}. Optimal results can be obtained by setting the initial value of $y_i^j$, $B_i^j$, $m_i^j$, $x_i^j$ and $P_i^j$, updating variables, and performing iterations until the problem converges. 

Finally, we can iteratively solve (\ref{Q4}) and (\ref{Q8}) until a convergent solution is obtained. The overall CRRA algorithm is summarized in \textbf{Algorithm} \ref{summary}.
\begin{algorithm}[t]
	\small
	\caption{CRRA Algorithm.}
	\begin{algorithmic}[1]
		\STATE Initialize semantic compression ratio $\boldsymbol o$, resource allocation $\boldsymbol B$ and $\boldsymbol P$.
		\REPEAT
		\STATE With fixed resource allocation $\boldsymbol B$ and $\boldsymbol P$, optimize semantic compression ratios $\boldsymbol o$ with the enumeration method.
		\STATE With fixed semantic compression ratios, obtain the optimal resource allocation $\boldsymbol B$ and $\boldsymbol P$ by solving (\ref{Q8}).
		\UNTIL {the objective value (\ref{Q3}) converges.}
	\end{algorithmic}
	\label{summary}
\end{algorithm}

\vspace{-0.2cm}
\section{Simulation Results and Analysis}
\label{sec:simulation}
In this section, we compare the proposed CRRA algorithm with three baselines: resource allocation scheme with fixed compression ratios (labeled as "FCR"), compression ratios optimization scheme with fixed resource allocation (labeled as "FRA"), and conventional resource allocation scheme to maximize the system sum rate (labeled as "MSR"). There are $U=10$ users uniformly distributed in a $R=50$ m square area with an edge sever at the center. The simulation parameters are summarized in Table \ref{tab:simulation}.
\begin{table}[htbp]
	\small
	\vspace{-0.1cm}
	\centering
	\caption{SIMULATION PARAMETERS.}
	\setlength{\abovecaptionskip}{-0.5cm}
	\begin{tabular}{ccc}
		\toprule
		Parameter & Value \\
		\hline
		Initial data size, ${{d_0}}$ & 24.5 MB \\
		Delay constraint of user $n$, $t_n$ & 1-10 ms\\
		Noise power spectral density, $N_0$ & -174 dBm/Hz\\
		Minimum bandwidth, ${{B_{\min}}}$ & 0.01 MHz\\
		Minimum transmit power, ${P_{\min}}$ & -20dBm\\
		The number of users, $U$ & 10\\
		Compression ratio, $o$ & 0-1\\
		Maximum bandwidth, $B_{\max}$ & 1-30 MHz\\
		Maximum transmit power, $P_{\max}$ & 1 mW-1 W\\
		\toprule
	\end{tabular}
	\vspace{-0.cm}
	\label{tab:simulation}
\end{table}

%

\begin{figure}[htbp]
	\begin{center}
		\includegraphics[width=1\linewidth]{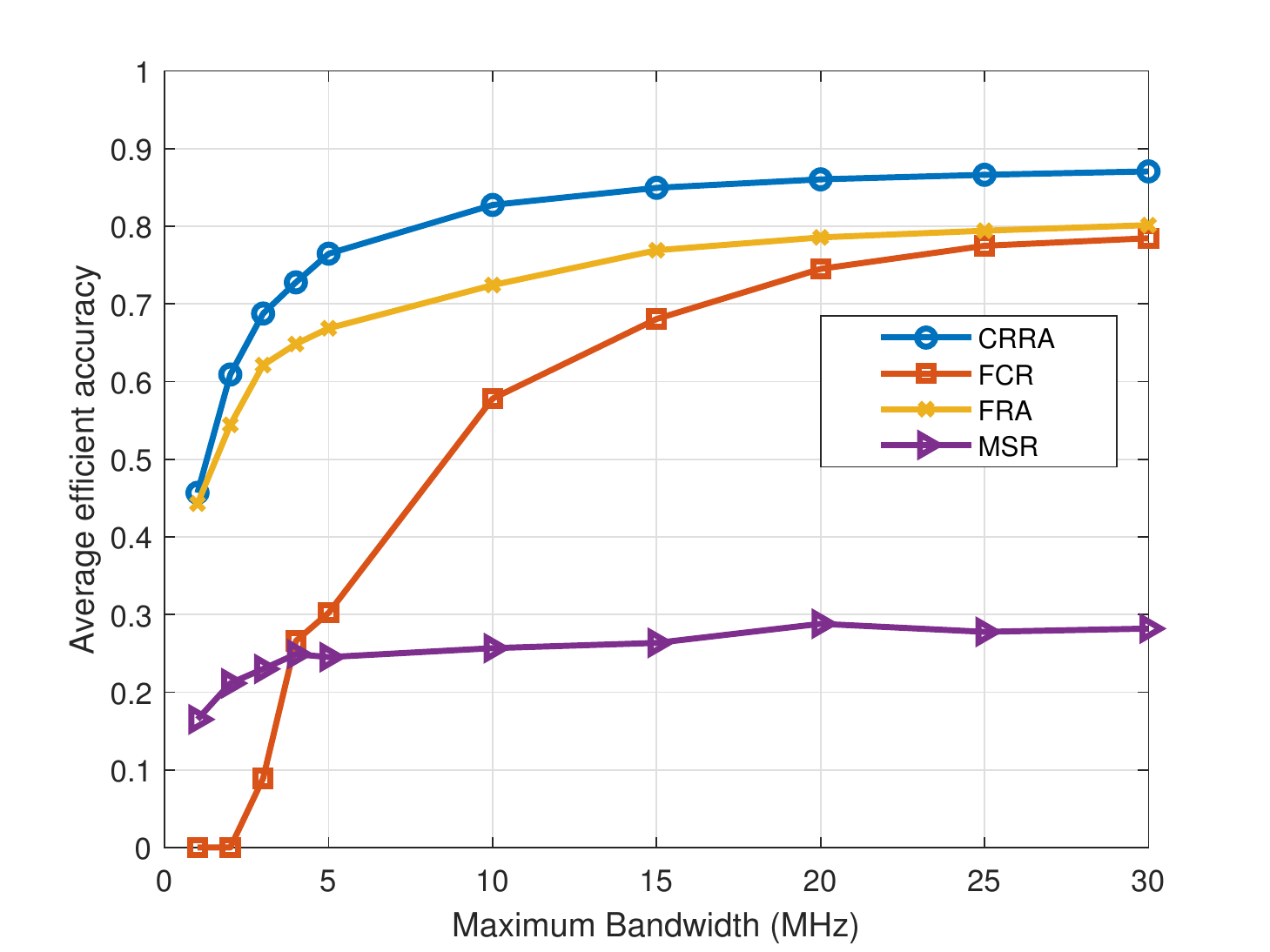}
	\end{center}
	\caption{Average effective accuracy versus the maximum bandwidth.}
	\vspace{-0.5cm}
	\label{fig:bandwidth1}
	\end{figure}

The average effective accuracy versus the maximum bandwidth are shown in Fig. \ref{fig:bandwidth1}. As shown in this figure, the average effective accuracy increases with the maximum bandwidth and gradually converges to a certain threshold. This is because large bandwidth can decrease the transmission delay and tolerate small semantic compression ratio, which consequently increases the probability of successful transmission and average effective accuracy. It can be observed that the average effcient accuracy of the proposed algorithm is always higher than that of others, especially in low bandwidth regions. It can be found that the conventional resource allocation scheme MSR is no longer suitable for semantic communication scenarios. This is because the conventional resource allocation scheme only optimizes the transmission rate and lacks the consideration of semantics and subsequent intelligent tasks.


\begin{figure}[htbp]
	\begin{center}
		\includegraphics[width=1\linewidth]{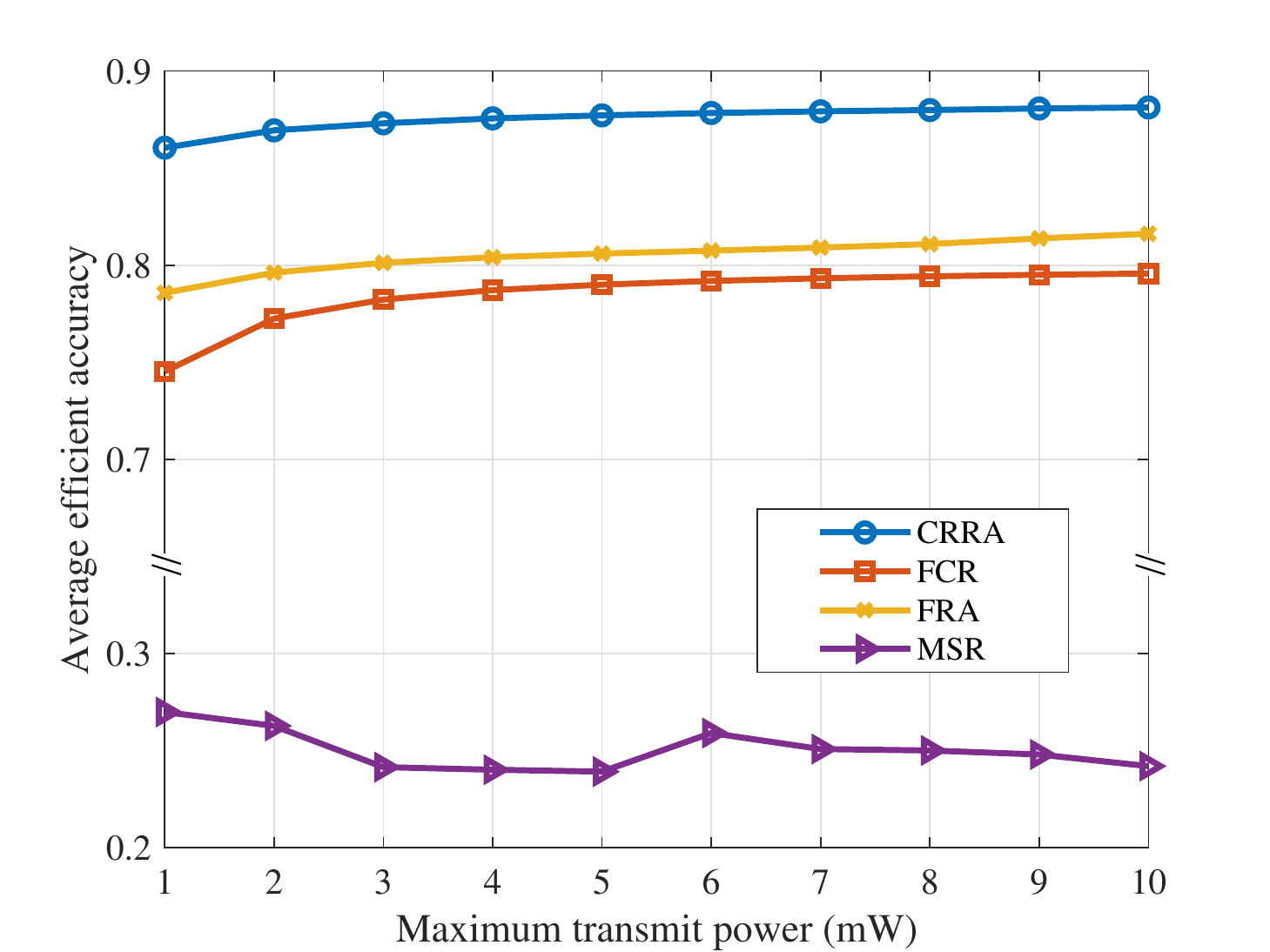}
	\end{center}
	\caption{Average effective accuracy versus the maximum sum transmit power.}
	\label{fig:power}
\end{figure}
The average effective accuracy versus the maximum transmit power is depicted in Fig. \ref{fig:power}. From this figure, we can observe that the proposed algorithm achieves better performance than FCR, FRA and MSR. Fig. \ref{fig:power} demonstrates that the average effective accuracy increases as the maximum transmit power. This is because large transmit power can increase the transmission rate, which consequently increases the amount of transmitted data. It can also be observed that the proposed algorithm harvests significant performance gains compared with the conventional schemes even under large transmit power. From Fig. \ref{fig:power}, we can further find that the conventional resource allocation method has little improvement in semantic performance. This is because that conventional method only focus on technical performance, which may not necessarily transmit the semantic information required for intelligent tasks well. Besides, the proposed algorithm can perform well even in very low transmit power regions, which shows that our algorithm is very suitable for low-power scenarios.

\vspace{-0.1cm}
\section{Conclusion}
\label{conclusion}

In this paper, we have investigated the AI tasks' performance maximization problem in semantic communication system via optimizing the resource allocation scheme. We have derived a closed-form expression of the relationship between semantic compression ratio and task performance and proposed the CRRA algorithm. The proposed algorithm takes the subsequent AI tasks into consideration and focus on the performance of semantic transmission, which differ from conventional resource allocation scheme. The complexity analysis of CRRA is omitted here due to space reasons and will be given in journal version. Simulation results have shown the superiority of the proposed algorithm.
%
\bibliographystyle{IEEEbib}
\nocite{*}\bibliography{stimreference}

\end{document}